\documentclass[11pt,fleqn]{article}
\usepackage{amsmath,amssymb,amsthm,cite,enumerate,graphics,epsfig,nth,mathrsfs}

\newtheorem{theorem}{Theorem}
\newtheorem{lemma}{Lemma}

\newtheorem{proposition}{Proposition}
{ \theoremstyle{definition}
\newtheorem{definition}{Definition}
\newtheorem{example}{Example}
\newtheorem{remark}{Remark}
}

\newcommand{\todo}[1][\null]{\ensuremath{\clubsuit}}

\newcommand{\R}{{\mathbb R}}
\newcommand{\Z}{{\mathbb Z}}
\newcommand{\N}{{\mathbb N}}
\newcommand{\C}{{\mathbb C}}
\newcommand{\g}{\mathfrak{g}}
\newcommand{\h}{\mathfrak{h}}
\newcommand{\id}{\mathrm{id}}

\newcommand{\e}{\mathrm{e}}
\newcommand{\RR}{\mathscr{R}}
\renewcommand{\SS}{\mathscr{S}}
\newcommand{\TT}{\mathscr{T}}

\newcommand{\Inn}{\mathop{\rm Inn}\nolimits}
\newcommand{\Aut}{\mathop{\rm Aut}\nolimits}
\newcommand{\Der}{\mathop{\rm Der}\nolimits}
\newcommand{\Vect}{\mathop{\rm Vect}\nolimits}
\newcommand{\rank}{\mathop{\rm rank}\nolimits}
\newcommand{\Dom}{\mathop{\rm Dom}\nolimits}
\newcommand{\Ad}{\mathop{\rm Ad}\nolimits}
\newcommand{\ad}{\mathop{\rm ad}\nolimits}
\newcommand{\Span}{\mathop{\rm span}}
\newcommand{\Gr}{\mathop{\rm Gr}\nolimits}
\newcommand{\GL}{\mathop{\rm GL}\nolimits}
\newcommand{\sgn}{\mathop{\rm sgn}\nolimits}
\newcommand{\rc}{\mathbin{\backslash}}

\newcommand{\D}{\partial}

\allowdisplaybreaks[1]

\begin{document}

\begin{center}
\LARGE \bf
On classification of Lie algebra realizations
\end{center}

\begin{center}
Daniel Gromada$^\dag$,
Severin Po\v sta$^\S$
\end{center}
\noindent $^\dag$~Department of Physics, Faculty of Nuclear Sciences and Physical Engineering,\\$\phantom{^\dag}$~Czech Technical University in Prague, B\v rehov\'a 7, CZ-115 19 Prague, Czech Republic\\
$\phantom{^\dag}$~E-mail: daniel.gromada@fjfi.cvut.cz

\noindent $^\S$~Department of Mathematics, Faculty of Nuclear Sciences and Physical Engineering,\\$\phantom{^\dag}$~Czech Technical University in Prague, Trojanova 13, CZ-120 00 Prague, Czech Republic\\
$\phantom{^\dag}$~E-mail: severin.posta@fjfi.cvut.cz

\date{\today}

\begin{abstract}
We study realizations of Lie algebras by vector fields. A correspondence between classification of transitive local realizations and classification of subalgebras is generalized to the case of regular local realizations. A reasonable classification problem for general realizations is rigorously formulated and an algorithm for construction of such classification is presented.
\end{abstract}


\section{Introduction}

There are several main classification problems in the theory of Lie algebras.
In particular description of Lie algebra representations by vector fields give raise to two classification problems. 
The first one was first considered by S.~Lie and is that of classifying, up to local diffeomorphisms, 
the finite-dimensional Lie algebras of vector fields defined on an open subset of the finite-dimensional (usually dimension is rather low 1, 2 or 3) Euclidean space, see, e.g.~\cite{lie1883} and \cite{Olver1}.
The second classification problem starts from the fixed structure of Lie algebra (it's commutation relations) and 
classify all it's inequivalent \emph{realizations}, i.e. representations by finite-dimensional vector fields.
There is a number of publications devoted to construction or classification of certain types of realizations or realizations of fixed Lie algebras.
 One of the most general result was obtained by Popovych et al. in \cite{Popovych2003} (see also references therein), where local realizations of all Lie algebras of dimension less or equal to four are constructed, although the actual classification problem was not specified in all details.

Investigation of realizations is motivated by a wide range of applications in the general theory of differential equations,
integration of differential equations and their systems~\cite{Olver0, Ovsyannikov},
in group classification of ODEs and PDEs~\cite{basarab-horwath&lahno&zhdanov2001},
in classification of gravity fields of a general form with respect to motion groups~\cite{petrov1966},
in geometric control theory and in the theory of systems
with superposition principles~\cite{Carinena,Shnider_Winternitz1984}.
Realizations are also applicable in the difference schemes for
numerical solutions of differential equations~\cite{BourliouxCyr-GagnonWinternitz}.
Description of realizations is the first step for solving the Levine's problem~\cite{Levine}
on the second order time-independent Hamiltonian operators which lie in the universal enveloping algebra of
a finite-dimensional Lie algebra of the first-order differential operators.
The Levine's problem was posed in molecular dynamics.
In such a way, realizations are relevant in the theory of quasi-exactly solvable problems of quantum mechanics through
the so-called algebraic approach to scattering theory and molecular dynamics and the list of possible applications of realizations of Lie algebras
is not exhausted by the above-mentioned subjects.

Usually, realizations are considered only locally due to the general locality of Lie approach. 
The research is often focused only on certain types of realizations such as transitive or regular since the classification problem is much more simple in those cases.

Theoretical results on classification of \emph{transitive} local realizations, together with powerful methods of explicit computation are already available in the literature\cite{Guillemin1964,Blattner1969,shirokov2013,draisma2012}.
In particular, it can be shown that classification of all transitive realizations is equivalent to classification of subalgebras. These results are summarized in Section \ref{sec.trans}.
The key correspondence between realizations and subalgebras was generalized to the case of regular realizations in Section \ref{sec.nontrans}.
In this section the main Theorem \ref{T.reg} is equipped by a simple example illustrating the practical computation.


When dealing with general realizations it is actually not so clear how to choose the classification problem to get reasonable results. As we show in Section \ref{sec.problem} one of possible classification problems could be to find a system of local realizations such that every realization is at each point of a dense subset of its domain locally equivalent to some realization of this system. Such a system will be called a \emph{complete system of local realizations} (see Definition \ref{D.complete}) and
we show that such a system can be chosen to consist of \emph{regular} local realizations.

Finally, in Section \ref{sec.complete}, we formulate an algorithm that solves the established classification problem (that is, to find the complete system of local realizations) and Theorem \ref{T.main} summarizing this algorithm is followed by an example illustrating this procedure.

\section{Realizations}
\label{sec.real}
Let $\g$ be a real $n$-dimensional Lie algebra with a
structure constants tensor $C_{ij}^k$, $i,j,k=1,\dots,n$ and $M$ an $m$-dimensional smooth manifold. 
In this paper we study realizations of Lie algebras by vector fields. By a \emph{vector field} we mean an element of derivation of a Lie algebra of smooth functions (that is, differentiable infinitely many times) on some manifold $M$. The Lie algebra of vector fields will be denoted $\Vect M$.

\begin{definition}
A \emph{realization} of $\g$ on the manifold $M$ is a homomorphism $R\colon \g\to\Vect M$. The realization is called \emph{faithful} if it is injective.
\end{definition}

Another approach is to consider a derivation of the algebra of formal power series over a field $F$ of characteristic zero $\Der F[[x]]$ instead of vector fields (see for example \cite{draisma2012}). This definition is more general regarding the arbitrary field $F$, on the other hand it corresponds to local analytic realizations only.

By restricting the realizing vector fields on an open subset $U$ of $M$ we get a \emph{restriction} of a realization on $U$, which we will denote $R|_U$. The manifold, where the realizing vector fields are defined is called the \emph{domain} of the realization and denoted $\Dom R$.

The realizations are often considered only locally. This means that we specify a point in the manifold $M$ and consider the realizing vector fields only in a small neighborhood of this point. This is equivalent to considering the realizations on a neighborhood of zero at $\R^m$. Formally, we can define a local realization as an equivalence class of realizations that coincide in some neighborhood of a given point.

\begin{definition}
\label{D.local}
Let a point $p\in M$ and let $U_1$ and $U_2$ be neighborhoods of $p$. Then realizations $R_1$ and $R_2$ of $\g$ defined on $U_1$ and $U_2$, respectively, \emph{locally coincide} at $p$ if there is a neighborhood $V\subset U_1\cap U_2$ of $p$ such that $R_1|_V=R_2|_V$. The classes of locally coincident realizations at a specified point $p$ are called {\em local realizations at $p$}. A local realization is \emph{faithful} if every its representative is faithful.
\end{definition}

We will usually not strictly distinguish between local realizations and their representatives. For a given global realization $R$ defined on a manifold $M$ we denote $R|_p$ the corresponding local realization at $p\in M$. For a given local realization $R$ at $p\in M$ and $U$ a neighborhood of $p$ we denote $R|_U$ the corresponding representative defined on $U$.

The fundamental result of the Lie theory is that local transformations are completely defined by vector fields representing the infinitesimal transformations. The local and infinitesimal transformations were later given an abstract structures of a (local) Lie group and Lie algebra. The actual transformation is then expressed as a (local) action of the Lie group on a manifold and the infinitesimal transformation is described by fundamental vector fields, that is a realization of the Lie algebra by vector fields on the manifold. So, in the modern language the Lie theory states that there is a one-to-one correspondence between local Lie group actions and realizations of a Lie algebra.

The notion of a local realization is introduced to simplify the classification problem leaving aside the global structure of the manifold and the realizing vector fields. In terms of Lie group action it would correspond to something we are going to call \emph{locally defined action}---an action that is local not only in the group variable but it is also defined only in some neighborhood of a given point of some manifold. Formally, it could be again defined as an equivalence class of local realizations that coincide in a given point.

\section{Equivalence of realizations}
\label{sec.equiv}

To state the classification problem for realizations, we present a definition of equivalence.

\begin{definition}
\label{D.equiv}
Let $\g$ be a Lie algebra and $M_1$ and $M_2$ manifolds. Let $A$ be a subgroup of $\Aut(\g)$. Realizations $R_1\colon\g\to\Vect(M_1)$ and $R_2\colon\g\to\Vect(M_2)$ are called {\em $A$-equivalent} if there exist an automorphism $\alpha\in A$ and a diffeomorphism $\Phi\colon M_1\to M_2$ such that $R_2(\alpha(v))=\Phi_*\,R_1(v)$ for all $v\in\g$. If we do not consider automorphisms (so $A=\{\id\}$) the realizations are called just equivalent (or {\em strongly} equivalent in case we need to emphasise that we do not consider automorphisms). For inner automorphisms $A=\Inn\g$ we will write shortly Inn-equivalent and for all automorphisms $A=\Aut\g$ we will say Aut-equivalent.
\end{definition}

\begin{definition}
\label{D.locequiv}
Local realization $R_1$ at $p_1\in M_1$ is $A$-equivalent to local realization $R_2$ at $p_2\in M_2$ if there exist their representatives defined in neighborhoods $U_1\owns p_1$, $U_2\owns p_2$ that are $A$-equivalent and the corresponding diffeomorphism $\Phi\colon U_1\to U_2$ satisfies $\Phi(p_1)=p_2$.
\end{definition}

The strong equivalence corresponds to isomorphism of the corresponding local actions. It also similar to the definition of equivalence in case of ordinary representations. On the other hand, the Aut-equivalence corresponds to similitude of actions.

For a given global or local realization $R$ we denote $\bar R$ the corresponding $A$-equivalence class. It~should be clear from context which group of automorphisms $A$ we consider.

Since faithfulness of the realization or the dimension of the realizing manifold is invariant under the equivalence, we can assign those characteristics to the classes. In particular, in case of local realizations, the global structure of the manifold is irrelevant and the dimension is the only characteristic of the manifold, so we often refer to a {\em class of local realizations in $m$ variables}.

Every class of local realizations in $m$ variables $\bar R$ of an $n$-dimensional Lie algebra $\g$ can be represented by a realization defined in a neighborhood of zero in $\R^m$, so it is determined by $n\cdot m$ functions $\xi_{ij}$, $i=1,\dots,n$, $j=1,\dots,m$ defined in some neighborhood of zero at $\R^m$ that form one of the representatives $R_0\in\bar R$
\begin{equation}
\label{eq.locreal}
R_0(e_i)=\sum_{j=1}^m\xi_{ij}(x_1,\dots,x_m)\partial_{x_j}.
\end{equation}
It means that for all representatives $R$ defined in a neighborhood of a point $p\in M$ there exist coordinates $(x_1,\dots,x_m)$ in some neighborhood of $p$ such that the coordinate expression of $R$ coincides with \eqref{eq.locreal}.

This also illustrates the connection with the definition of realizations by derivations of formal power series. In this case, we have formal power series instead of the functions $\xi_{ij}$ and the equivalence is provided by formal coordinate change preserving zero (since all the power series are centered at zero). Therefore, there is a one-to-one correspondence between local realizations at zero that have an analytic representative and formal realizations that have a convergent representative and also between the corresponding classes. In the case of so called transitive realizations (see Section \ref{sec.trans}), it can be shown that all classes of local realizations have an analytic representative and all classes of formal realizations have a convergent representative.

%

\section{Transitive local realizations}
\label{sec.trans}

\begin{definition}
Let $R$ be a realization of $\g$ on $M$. The {\em rank} of the realization $R$ at a point $p\in M$ is the rank of the linear map $R_p\colon\g\to T_pM$, $v\mapsto R(v)_p$, that is, $\rank R_p=\dim R(\g)_p=\dim\{R(v)_p\mid v\in\g\}$. If the function $p\mapsto\rank R_p$ is locally constant at $p_0$ we say that $R$ is {\em regular} at $p_0$. Local realizations at $p$ are called {\em regular} if their representatives are regular at $p$.
\end{definition}

It is clear that rank of a realization is always less or equal to the dimension of the manifold $m$. Realizations with maximal rank are called transitive.

\begin{definition}
A realization $R$ of $\g$ on $M$ is called \emph{transitive} if $\rank R_p=m$ for all $p\in M$. A local realization in $m$ variables at $p$ is called transitive if $\rank R_p=m$.
\end{definition}

This definition of transitivity of a (local) realization corresponds to the transitivity of a (local) action associated with the realization that is given by integrating the realizing vector fields. Classification of those local actions leads to classification of local transitive realizations.

Any transitive local realization at $p$ of a given Lie algebra $\g$ defines a subalgebra $\h=\ker R_p\subset\g$ as the kernel of the linear map $\g\to T_pM$, $v\mapsto R(v)_p$. This is obviously invariant with respect to coordinate change, so it does not depend on the representative chosen. The codimension of $\h$ equals to the rank of $R$. It can be proven that there is actually a one-to-one correspondence between strong classes of local realizations and subalgebras. This can be shown very easily using the theory of the corresponding local Lie group actions \cite{shirokov2013}.

First of all, take a local Lie group $G$ corresponding to the Lie algebra $\g$ and an action $\pi$ locally defined at $p$ corresponding to local realization $R$ at $p$. Then the kernel $\h$ of $R_p$ corresponds to the stabilizer $G_p$ of $\pi$. From the theory of a Lie group action on a manifold, we know that every right action of $G$ is isomorphic to right multiplication of $G$ on the space of right cosets $G_p\rc G$. So, we have the uniqueness---all realizations corresponding to a subalgebra $\h$ are equivalent. Now, the existence. Given an arbitrary subalgebra $\h\subset\g$ and taking the corresponding subgroup $H\subset G$, the right multiplication of $G$ on $H\rc G$ has a stabilizer $H$. Fundamental vector fields of this action then form the realization of $\g$ corresponding to the subalgebra $\h$.

This correspondence can be formulated purely algebraically without need of introducing (local) Lie groups and their action. In the case of realizations by formal power series over general field $F$, the correspondence was proven by Guillemin and Sternberg \cite{Guillemin1964}. Later, Blattner \cite{Blattner1969} came with even more abstract proof of correspondence between subalgebras and certain classes of representations. These works are valuable not only because they are very general in the definition of realization, but they are also obtained purely algebraically.

Since the subalgebra corresponding to a transitive local realization $R$ composed with an automorphism $\alpha\in A\subset\Aut\g$ is just $\alpha(\h)$, where $\h$ is the subalgebra corresponding to $R$, we conclude, that there is also one-to-one correspondence between $A$-classes of transitive local realizations and $A$-conjugacy classes of $\g$ subalgebras.

We are also able to characterize the faithfulness of the realization in terms of the subalgebra properties. From the theory of group actions we know that the kernel of an action of right multiplication on $H\rc G$ is the largest normal subgroup contained in $H$. Transferring these relations to Lie algebras and realizations we have the following propositions.

\begin{lemma}
\label{L.ideal}
The kernel of a local transitive realization of $\g$ is the largest ideal contained in the corresponding subgroup of $\g$.
\end{lemma}

\begin{lemma}
A transitive realization of a Lie algebra $\g$ is faithful if and only if the corresponding subalgebra of $\g$ does not contain any non-trivial ideal of $\g$.
\end{lemma}

This means that classification of transitive local realizations of a given Lie algebra with respect to a group of automorphisms $A$ is completely equivalent to classification of subalgebras with respect to $A$.

Moreover, the work of Shirokov et al. \cite{shirokov1997,shirokov2013} provides a very simple method for explicit computation of such realizations. The components of the realizing vector fields are computed from the structure constants using only matrix exponentiation and inversion and can be done completely automatically by a computer. This method was already used to classify transitive realizations of low-dimensional Poincar\'e algebras \cite{Nesterenko2015} and Galilei algebras \cite{Nesterenko2016}.

In short, the local transitive realization corresponding to a subalgebra $\h\subset\g$ of codimension $m$ is computed as follows. Choose a basis $(e_1,\dots,e_n)$ of $\g$ such that $(e_{m+1},\dots,e_n)$ is a basis of $\h$. Compute a matrix $\Omega$, whose elements are given by the following formula
$$\Omega(x_1,\dots,x_n)_{ij}=[\exp(-x_1\ad_{e_1})\cdot\exp(-x_2\ad_{e_2})\cdots\exp(-x_{j-1}\ad_{e_{j-1}})]_{ij}.$$
Then the first $m$ columns of the inverse matrix $\Omega^{-1}$ do not depend on $x_{m+1},\dots,x_n$ and form the realizing vector fields in $m$ variables $x_1,\dots,x_m$ in a neighborhood of zero
$$[R(e_j)_{x_1,\dots,x_m}]_i=[\Omega^{-1}(x_1,\dots,x_n)]_{ij},\quad i=1,\dots,m.$$

\subsection{Lie's conjecture}
In \cite{lie1893} Lie conjectured that any local transitive realization can be expressed (after a suitable change of coordinates) by entire functions of coordinates and exponentials of linear functions in coordinates (over $\C$).  Over $\R$ (or arbitrary other field) it can be reformulated as follows. Any local transitive realization can be expressed in certain coordinates (i.e. any class has such representative in $\R^m$) as functions of coordinates that are a solution of some differential equation with constant coefficients.

For certain types of realizations, this conjecture was proven by Draisma \cite{draisma2002} but it was not proven generally (and Draisma believes that it is generally not true). The Shirokov's metod does not prove or disprove this conjecture. Nevertheless, we are able to formulate weaker proposition. Every realization constructed by the Shirokov's method is a rational function of functions of coordinates that are a solution of a differential equation with constant coefficients. So, over $\C$ it would be a rational function of exponentials.

\section{Inner automorphisms}

In this section, we describe, how inner automorphisms act on transitive local realizations, which will be important for classification of regular local realizations.

\begin{lemma}
\label{L.globinn}
Let $G$ be a (global) Lie group and $\g$ its Lie algebra. Let $H$ be a subgroup of $G$ and take $g\in G$. Then (global) realization on $H\rc G$ by fundamental vector fields of right multiplication by $G$ is equivalent to realization by fundamental vector fields on the manifold $\tilde H\rc G$, $\tilde H=g^{-1}Hg$.
\end{lemma}
\begin{proof}
Take the action $\pi$ of right multiplication corresponding to the first realization. Then $H$ is the stabilizer of the class corresponding to unity $\bar e=H$. We can easily see that $\tilde H=g^{-1}Hg$ is the stabilizer of a point $\bar eg=Hg=\bar g$. At the same time, $\pi$ is isomorphic to the right multiplication on $\tilde H\rc G$, which is the action that corresponds to the second realization. This action isomorphism being also a manifold diffeomorphism provides the equivalence.
\end{proof}

This means that, globally, different subgroups or subalgebras can correspond to equivalent global realizations. In terms of local realizations, we can formulate the following proposition.

\begin{lemma}
\label{L.inn}
Let $R$ be a transitive realization of $\g$ on $M$, $p\in M$. Then for every $q\in M$ the local realizations $R|_p$ and $R|_q$ are equivalent with respect to inner automorphisms. Conversely, for every neighborhood $U$ of $p$ there exists a neighborhood of unity $V$ in the group of inner automorphisms such that for every $\alpha\in V$ there exists $q\in U$ such that $R|_p$ and $R|_q$ correspond to $\alpha$-conjugated subalgebres and hence are $\alpha$-equivalent.
\end{lemma}
\begin{proof}
Denote $\h$ and $\tilde\h$ the subalgebras corresponding to local realizations $R|_p$ and $R|_q$. Choose a local Lie group $G$ and denote $H$, $\tilde H$ the corresponding subgroups. Denote $\pi$ the corresponding action of $G$ on $M$. From transitivity of the realization there exist $g\in G$ such that $q=\pi(p,g)$, so $\tilde H=g^{-1}Hg$. Thus, $\tilde\h=\Ad_{g^{-1}}\h$, so the realizations are equivalent with respect to this inner automorphism.

The second proposition is just a local version of Lemma \ref{L.globinn}.
\end{proof}

\begin{example}
Take a Lie algebra $\g_{3,1}=\Span\{e_1,e_2,e_3\}$, $[e_2,e_3]=e_1$. All one-dimensional subalgebras and the corresponding realizations in a neighborhood of $(0,0)\in\R^2$ are following.
\begin{eqnarray*}
\Span\{e_1\}\colon&\qquad e_1\mapsto 0,\quad e_2\mapsto\partial_1,\quad e_3\mapsto\partial_2\\
\Span\{e_2-ae_1\}\colon&\qquad e_1\mapsto\partial_1,\quad e_2\mapsto (a-x_2)\partial_1,\quad e_3\mapsto\partial_2\\
\Span\{e_3-be_2-ae_1\}\colon&\qquad e_1\mapsto\partial_1,\quad e_2\mapsto\partial_2,\quad e_3\mapsto(a+x_2)\partial_1+b\partial_2
\end{eqnarray*}

All the subalgebras in the second row are equivalent with respect to inner automorphisms and, for fixed $b$, all the realizations in the third row are equivalent with respect to inner automorphisms. In other words, classification of subalgebras and realizations with respect to inner automorphisms is obtained by removing the parameter $a$ (setting $a:=0$).

All these realizations are, of course, mutually locally inequivalent at zero (they are representatives of different local realizations). However, if we consider them as global realizations on $\R^2$, then all realizations in the second row are equivalent and this equivalence is provided by simple translation in $x_2$. The same holds for the last row for fixed $b$.

\end{example}

%
%

\section{Topology of subalgebra and realization systems}
\label{sec.topology}

In the following section we are going to construct new realizations by interpreting parameters of transitive realizations as new coordinates. This is, of course, possible only in the case when the transitive realizations depend ``smoothly'' on those parameters.

We are going to say that a function $F\colon\R^s\to\Vect M$ is smooth if a vector field $X^F\in\Vect(M\times \R^s)$ defined as $X^F_{(p,x)}=F(x)_p$ for $p\in M$ and $x\in\R^s$ is smooth. As indicated for example in \cite{Munkers}, Theorem 46.11, such a definition corresponds to compact open topology on $\Vect M$. For further references on the vector fields topology, see e.g. \cite{Jafarpour}. Nevertheless, we will not use any special properties of such topology here.

This induces a topology and the notion of smoothness on the space of realizations. We also get a topology on the space of local realizations and spaces of $A$-classes of local realizations as a topological quotient spaces. Note, however, that those quotient spaces may not be even Hausdorff.

Nevertheless, we can again induce the notion of smooth map. A map between two quotient spaces will be called smooth if it is locally a quotient of a smooth map. We formulate it precisely in the following definition.

\begin{definition}
\label{D.smooth}
Let $M_1$ and $M_2$ be manifolds and let $\bar M_1$ and $\bar M_2$ be their quotient spaces. A map $\bar\Phi\colon\bar M_1\to\bar M_2$ will be called \emph{smooth} in $\bar x_0\in\bar M_1$ if there exists $x_0\in\bar x_0$, its neighborhood $U$, and a map $\Phi\colon U\to M_2$ such that for all $x\in U$ $\Phi(x)\in\bar\Phi(\bar x)$, where $\bar x$ is the class corresponding to $x$. A smooth bijection, whose inversion is smooth as well, will be called a \emph{diffeomorphism}. A smooth injection, whose inversion is smooth as well, will be called an \emph{embedding}.
\end{definition}

\begin{remark}
This definition is compatible with the quotient topology in a sense that every smooth map is continuous, so every diffeomorphism is a homeomorphism. A composition of such smooth maps is smooth.
\end{remark}

Let $\SS_m$ be the space of all subalgebras of codimension $m$ of a given Lie algebra $\g$. This is an affine subvariety of the Grassmannian $\Gr(\g,n-m)$, $n=\dim\g$. A map $F\colon\R^s\to\SS_m$ will be called smooth if it is smooth as a map $\R^s\to\Gr(\g,n-m)$. A smooth map to the space of subalgebra $A$-classes is again defined in sense of Definition \ref{D.smooth}.

\begin{lemma}
Let $\g$ be a Lie algebra, $\bar\SS_m$ the space of all $A$-classes of subalgebras of codimension $m$ in $\g$ and $\bar\TT_m$ the space of all $A$-classes of local transitive realizations in $m$ variables. Then $\bar\SS_m$ is diffeomorphic to $\bar\TT_m$.
\end{lemma}
\begin{proof}
It is sufficient to prove this proposition for space of subalgebras $\SS_m$ and space of strong classes $\bar\TT_m$. Then we only ``factor'' both sides.

In Section \ref{sec.trans} we showed that there is a bijection between these two sets. It is clear that the map $\TT_m\to\SS_m$, $R\mapsto\ker R_p$ is smooth. Therefore, the same holds for the map of classes. To show the smoothness of the inverse, we can make use of the Shirokov's computation that also smoothly depends on the choice of the subalgebra.
\end{proof}

Therefore, the space of strong classes of transitive realizations actually is an homeomorphic image of an algebraic variety and therefore it is a Hausdorff space. The space of general $A$-classes, however, does not have to be.


\section{Regular local realizations}
\label{sec.nontrans}

In this section we characterize the classification problem for regular realizations. Classification of regular realizations is very important also for general realizations since we have the following lemma.

\begin{lemma}
\label{L.regdense}
Let $R$ be a realization of $\g$ on $M$. Then the regular points of $R$ form an open dense set in $M$. Hence, the set of singular points is nowhere dense.
\end{lemma}
\begin{proof}
The set is obviously open. Choose a point $p_0\in M$ and its neighborhood $U$. We find a regular point $p\in U$.

Denote $r:=\max_{p\in U}\rank R_p$. We can easily construct a continuous function $f\colon U\to\R$ such that $\rank R_p=r$ if and only if $F(p)\neq 0$ as sum of squares of some minors of the linear map $R_p$. Preimage of $\R\setminus 0$ is open, nonempty, and contain points, where $\rank R_p$ is locally constant and equal to $r$.
\end{proof}

Now, we generalize the correspondence between transitive realizations and subalgebras to the case of regular realizations.

Let us have a local realization $R$ at $p\in U$ with constant rank $r$ on $U$. Then $R(\g)$ forms an involutive $r$-dimensional distribution so, according to the Frobenius theorem, we can choose coordinates $x_1,\dots,x_m$ on a neighborhood $U$ of $p$, such that $U$ is foliated by integral submanifolds given by equations $x_i={\rm const}$ for $i=r+1,\dots,m$. These integral submanifolds are also the orbits of the action $\pi$ corresponding to the realization $R$. The basis elements $e_1,\dots,e_n$ of $\g$ are, therefore, realized by vector fields of the form
$$R(e_i)=\sum_{j=1}^{r}\xi_{ij}(x_1,\dots,x_m)\partial_{x_j}.$$

This realization induces an $(m-r)$-parameter set of realizations parametrized by $x_{r+1},\dots,x_m$ on the submanifolds $pG\simeq G_p\rc G$. These realizations are transitive, so they are equivalent to the realizations found by the algorithm described in Section \ref{sec.trans}.

Note that the action on the whole manifold uniquely defines transitive action on the orbit of a given point. So, a regular realization $R$ in a given point defines a unique transitive realization on the integral submanifold of the point. We will call this realization a {\em transitive restriction} of $R$. This relation obviously does not break by applying a diffeomorphism or an automorphism. A local diffeomorphism of the whole neighborhood $U$ induces a diffeomorphism of the orbits. An orbit of an action does not change by composing it with an automorphism of the group.


To find all non-transitive regular realizations, we can proceed the other way around. Every local regular realization in $m$ variables is of the form
$$R(e_j)_{x_1,\dots,x_m}=R^{(x_{r+1},\dots,x_m)}(e_j)_{x_1,\dots,x_r},$$
where $R^{(a_1,\dots,a_{m-r})}$ is $m-r$ parameter set of transitive realizations. 

Now, we formulate the main theorem. In the formulation we use the term {\em local smooth $s$-parameter set of subalgebra classes with codimension $r$}. By that we mean a smooth map $U\to\bar\SS_r$, where $U$ is a neighborhood of $0\in\R^s$ considering its values only locally (as in Def. \ref{D.local}). By a {\em class} of such maps we mean a class of equivalence up to ``regular reparametrization'', that is, $S$ and $S'$ are equivalent if and only if there exists a local diffeomorphism $\Psi\colon\R^s\to\R^s$, $\Psi(0)=0$ such that $S'=S\circ\Psi$. 

\begin{theorem}
\label{T.reg}
Let $\bar\SS_r$ the system of all Inn-classes of subalgebras with codimension~$r$. For $\bar\h\in\bar\SS_r$ denote $\bar R^{\bar\h}$ the corresponding Inn-class of transitive realizations. Then there is a bijection between Inn-classes of regular local realizations of $\g$ in $m$ variables with rank $r$ and classes of local smooth $(m-r)$-parameter sets of Inn-classes of subalgebras with codimension $r$. For any such $(m-r)$-parameter set $S\colon V\to\bar\SS_r$, we define a local realization $R\in\Vect(U\subset\R^m)$ at zero as follows
\begin{equation}
\label{eq.Treg}
R(e_j)_{x_1,\dots,x_m}=R^{S(x_{r+1},\dots,x_m)}(e_j)_{x_1,\dots,x_r},
\end{equation}
where the representatives $R^{S(x_{r+1},\dots,x_m)}$ are local realizations at $0\in\R^r$ chosen to be smooth in the variables $x_{r+1},\dots,x_m$.
\end{theorem}
\begin{proof}
At first, we prove that the map is well-defined. The smoothness of $S$ implies that we have indeed defined a smooth vector fields in sense of Section \ref{sec.topology}. Next, we have to show that those vector fields do not depend on the choice of representative $S$ and representatives of the realizations $R^S$.

Assume we chose another representatives for both $R^S$ and $S$, say
$$R'^{(S\circ\Psi)(x_{r+1},\dots,x_m)}=\Phi^{(x_{r+1},\dots,x_m)}_*R^{(S\circ\Psi)(x_{r+1},\dots,x_m)}\circ\alpha,$$
where $\Phi^{(x_{r+1},\dots,x_m)}$ is a smoothly parametrized set of diffeomorphisms. Then the resulting realization would be
$$R'(e_j)_{x_1,\dots,x_m}=\Phi^{(x_{r+1},\dots,x_m)}_*R^{(S\circ\Psi)(x_{r+1},\dots,x_m)}(\alpha(e_j))_{x_1,\dots,x_r}=\tilde\Phi_*R(\alpha(e_j))_{x_1,\dots,x_m},$$
where $\tilde\Phi\colon\R^m\to\R^m$ is a local diffeomorphism defined as
$$\tilde\Phi(x_1,\dots,x_m)=(\Phi^{(x_{r+1},\dots,x_m)}(x_1,\dots,x_r),\Psi(x_{r+1},\dots,x_m)).$$

The surjectivity of such a map follows from the Frobenius theorem as was described above.

To prove the injectivity, let us assume that realizations $R_1$ and $R_2$ of the form \eqref{eq.Treg} corresponding to local maps $S_1$ and $S_2$ are equivalent so $\Phi_*R_1=R_2\circ\alpha$.

The diffeomorphism $\Phi$ must preserve the integral submanifolds $x_j={\rm const}$ for $j>r$, so $\Phi^j(x_1,\dots,x_r,x^0_{r+1},\dots,x^0_m)$ has to be constant in $x_1,\dots,x_r$ for $j>r$. Hence, we can denote
$$\Phi(x_1,\dots,x_m)=
\begin{pmatrix}
\tilde\Phi_1^{(x_{r+1},\dots,x_m)}(x_1,\dots,x_r)\\\Phi_2(x_{r+1},\dots,x_m)
\end{pmatrix},$$
where $\tilde\Phi_1^{(x_{r+1},\dots,x_m)}\colon\R^r\to\R^r$ is an $(m-r)$-parameter set of local diffeomorphisms and $\Phi_2$ is a local diffeomorphism of $\R^{m-r}$. Note that although $\tilde\Phi_1^{(0,\dots,0)}(0,\dots,0)$ equals zero, generally $\tilde\Phi_1^{(x_{r+1},\dots,x_m)}(0,\dots,0)$ does not have to be zero, so these local diffeomorphisms at zero translate the point zero. So, denote
$$\Phi_1^{(x_{r+1},\dots,x_m)}(x_1,\dots,x_r):=\tilde\Phi_1^{(x_{r+1},\dots,x_m)}(x_1,\dots,x_r)-\tilde\Phi_1^{(x_{r+1},\dots,x_m)}(0,\dots,0).$$
We can write
$$\Phi_*R_1(e_j)_{x_1,\dots,x_m}=\Phi_{1\,*}^{(x_{r+1},\dots,x_m)}R^{(S_1\circ\Phi_2)(x_{r+1},\dots,x_m)}(\alpha^{(x_{r+1},\dots,x_m)}(e_j))_{x_1,\dots,x_r},$$
where $\alpha^{(x_{r+1},\dots,x_m)}$ is the inner automorphism corresponding (in sense of Lemma \ref{L.inn}) to translation $\tilde\Phi_1^{(x_{r+1},\dots,x_m)}(0,\dots,0)$. So, the equivalence means that
$$\Phi_{1\,*}^{(x_{r+1},\dots,x_m)}R^{(S_1\circ\Phi_2)(x_{r+1},\dots,x_m)}\circ\alpha^{(x_{r+1},\dots,x_m)}=R^{S_2(x_{r+1},\dots,x_m)}\circ\alpha,$$
so $R^{(S_1\circ\Phi_2)(x_{r+1},\dots,x_m)}$ is equivalent to $R^{S_2(x_{r+1},\dots,x_m)}$, which holds if and only if the corresponding classes of subalgebras coincide, so $S_1\circ\Phi_2=S_2$, which means that $S_1$ is equivalent to $S_2$.
\end{proof}

\begin{remark}
\label{R.outer}
The meaning of the map $S$ is that for a realization $R$ of the form \eqref{eq.Treg}, the Inn-class of subalgebras $S(x_{r+1},\dots,x_m)$ corresponds to the transitive restriction of $R$ in $(x_1,\dots,x_m)$, where $x_1,\dots,x_r$ are arbitrary and determine only the class representative. Therefore, outer automorphisms act on $R$ only by modifying these subalgebra classes. For $\alpha\in\Aut\g$ we have
$$R(\alpha(e_j))_{x_1,\dots,x_m}=R^{(\bar\alpha\circ S)(x_{r+1},\dots,x_m)}(e_j)_{x_1,\dots,x_r},$$
where $\bar\alpha\in\Aut\g/\Inn\g$ is the corresponding class of $\alpha$. Therefore, regular realizations corresponding to $S_1$ and $S_2$ are Aut-equivalent if and only if $S_1=\bar\alpha\circ S_2$ for some outer automorphism $\bar\alpha$.
\end{remark}

\begin{remark}
If the map $S$ is constant, the resulting realization is of the form
$$R(e_j)_{x_1,\dots,x_m}=R^{\h_0}(e_j)_{x_1,\dots,x_r},$$
so it has the same form as the original transitive realization. It is just formally defined on larger manifold. Such a regular realization will be called {\em trivial extension} of the transitive realization.
\end{remark}

\begin{lemma}
 A realization of the form \eqref{eq.Treg} is faithful if and only if the subalgebras in classes $S(x_{r+1},\dots,x_m)$ contain a common non-trivial ideal.
\end{lemma}
\begin{proof}
Follows from Lemma \ref{L.ideal}
\end{proof}

\begin{example}
\label{E.2g1reg}
Let us try to classify all regular realizations of two-dimensional Abelian Lie algebra $2\g_1=\Span\{e_1,e_2\}$. In this case, there are no non-trivial inner automorphisms, so Inn-equivalence is the same as strong equivalence. There is of course only one zero-dimensional subalgebra that corresponds to transitive realization
$$e_1\mapsto\partial_1,\quad e_2\mapsto\partial_2.$$
Since the set $\SS_2$ contains only one element, all regular realizations with rank two are obtained as trivial extension of this realization.

The set of one-dimensional subalgebras form a circle and can be parametrized by $\phi\in[0,2\pi)$ as $\h_\phi=\Span\{\cos\phi\,e_1+\sin\phi\,e_2\}$. The corresponding transitive realizations are following
$$R^{\h_\phi}(e_1)_{x_1}=\sin\phi\,\partial_1,\quad R^{\h_\phi}(e_2)_{x_1}=\cos\phi\,\partial_1.$$
Therefore, the set of all regular local realizations with rank one consists of following realizations
$$R^f(e_1)_{x_1,\dots,x_m}=\sin f(x_2,\dots,x_m)\,\partial_1,\quad R^f(e_1)_{x_1,\dots,x_m}=-\cos f(x_2,\dots,x_m)\,\partial_1,$$
where $f\colon\R^{m-1}\to\R$ are arbitrary functions. Such realizations $R^f$ and $R^g$ are equivalent if and only if there exists a local diffeomorphism $\Psi\colon\R^{m-1}\to\R^{m-1}$ at zero such that $g=f\circ\Psi$.

Now let us try to do this classification with respect to all automorphisms. We only have to describe, when $R^f$ and $R^g$ are Aut-equivalent. In this case, $\Aut(2\g_1)$ consists of all invertible linear maps. One such map is rotation, which acts as a rotation also on the circle $\SS_1=\{\h_\phi\}$. Thus, realizations $R^f$ and $R^g$, where $f$ and $g$ differ by a constant, are Aut-equivalent. We can therefore fix, for example, $f(0,\dots,0)=0$.

For a general automorphism $\alpha\in\Aut(2\g_1)=\GL(2\g_1)$, the subalgebra $\alpha(\h_{f(x_2,\dots,x_m)})$ is generated by
$$\begin{pmatrix}\alpha_{11}&\alpha_{12}\\\alpha_{21}&\alpha_{22}\end{pmatrix}\begin{pmatrix}\cos f(x_2,\dots,x_m)\\\sin f(x_2,\dots,x_m)\end{pmatrix}$$
Such an automorphism preserves the property $f(0,\dots,0)=0$ if and only if $\alpha_{11}=1$ and $\alpha_{21}=0$. Fixing these entries, $\alpha$ is an automorphism if and only if $\alpha_{22}\neq 0$.

So, the conclusion is that every regular local realization with rank one is Aut-equivalent to $R^f$, where $f(0,\dots,0)=0$. Such realizations $R^f$ and $R^g$ are mutually Aut-equivalent if and only if
$$\cos (g\circ\Psi)=\frac{\cos f+\alpha_{12}\sin f}{(\cos f+\alpha_{12}\sin f)^2+(\alpha_{22}\sin f)^2},$$
$$\sin (g\circ\Psi)=\frac{\alpha_{22}\sin f}{(\cos f+\alpha_{12}\sin f)^2+(\alpha_{22}\sin f)^2},$$
where $\alpha_{12}, \alpha_{22}\in\R$, $\alpha_{22}\neq 0$, $\Psi$ is a local diffeomorphism in $m-1$ variables.
\end{example}

\begin{example}
One should pay attention to the fact that the space $\bar\SS_m$ does not have to be Hausdorff. Take a non-commutative two-dimensional Lie algebra $\g_2=\Span\{e_1,e_2\}$, $[e_1,e_2]=e_1$. All one-dimensional subspaces of the form $\h_2^a:=\Span\{e_2+ae_1\}$ are equivalent to $\h_2^0=\Span\{e_2\}$ with respect to inner automorphisms. Therefore, there are only two Inn-classes of one-dimensional subalgebras represented by $\h_1:=\Span\{e_1\}$ and $\h_2^0$. This, however, does not mean that there are no non-constant smooth curves $S\colon\R\to\bar\SS_1$. We can take, for example, map $S(x)=\overline{\Span\{e_1+xe_2\}}$ that is evidently smooth despite it is equal to $\bar\h_1$ in zero and $\bar\h_2^0$ everywhere else. Therefore, it leads to a new regular local realization not equivalent to trivial extension of the transitive ones.
\end{example}

%
%

Note that in \cite{Nesterenko2016}  authors suggested replacing the parameters by new variables (not functions of variables), but it was not discussed what kind of new realizations are obtained or whether the list of realizations is complete.

\section{Classification problem}
\label{sec.problem}

In this section, we are going to discuss, what is the reasonable classification problem for (possibly general) Lie algebra realizations and what was the classification problem solved in \cite{Popovych2003}.

Sets of (local) realizations will be denoted by capital script letters $\RR,\TT,\dots$. The sets of corresponding $A$-classes will be denoted with bar $\bar\RR=\{\bar R\mid R\in\RR\}$.

Let $\bar\RR_{\rm all}$ be the system of all classes of local realizations of a given Lie algebra $\g$ with respect to a given group of automorphisms $A\subset\Aut\g$. A complete classification of local realizations would mean to find this system or, more precisely, to find a set of representatives $\RR_{\rm all}$ that would contain precisely one representative of every class in $\bar\RR_{\rm all}$. We prefer to choose representatives defined in a neighborhood at $0\in\R^m$. Every global realization would be at every point locally equivalent to a realization from our list. However, such classification would be very hard to perform. Nevertheless, the situation will get much more simple if we only require that every global realization is equivalent to a local realization from our list at every point \emph{from a dense subset}.

\begin{definition}
\label{D.complete}
Let $\g$ be a Lie algebra, $A\subset\Aut\g$ a group of its automorphisms. Let $\bar\RR$ be a system of local realizations classes of $\g$. We will say that $\bar\RR$ is {\em complete} if for every realization $R$ on any manifold $M$ there is a point $p\in M$ such that $\bar R|_p\in\bar\RR$.
\end{definition}

Note that this condition consequently means that for every realization $R$ on any manifold $M$ there is a dense subset of points $p\in M$ such that $\bar R|_p\in\bar\RR$ since the realization $R$ can be restricted to arbitrary open subset where the point has to exist as well.

We can reformulate this condition for the corresponding set of representatives. A set of local realizations $\RR$ is called {\em complete} if the corresponding system $\bar\RR$ is complete. It can be easily seen that such a set $\RR$ is complete if and only if and only if for every realization $R$ on any manifold $M$ there is a point $p\in M$ and a local realization $R'\in\RR$ such that $R|_p$ is $A$-equivalent to $R'$.

If we consider a complete system of local realizations $\RR$ such that all elements are defined in a neighborhood of zero in $\R^m$, then the completeness means following. For every realization $R$ on any manifold $M$ there exist a point $p\in M$ (in fact a dense set of such points), coordinates in a neighborhood of this point, and a realization $R_0\in\RR$ such that $p$ is in origin of these coordinates and the coordinate expression for $R$ coincides with $R_0$ (up to automorphisms).

\begin{example}
Let us take two dimensional non-commutative Lie algebra $\g_2=\Span\{e_1,e_2\}$, $[e_1,e_2]=e_1$. It can be shown that the complete system of local realizations in one variable of $\g_2$ can be chosen to contain only zero realization and
$$R(e_1)_x=\partial_x,\quad R(e_2)_x=x\partial_x.$$
\smallskip
In \cite{Spichak} Spichak classified realizations of $\g_2$ on circle. He used weaker definition of realization and equivalence, but we can use his results as an example of global realizations
$$R_n(e_1)_\theta=(\cos n\theta-1)\partial_\theta,\quad R_n(e_2)_\theta=\left(\cos n\theta-1-\frac{1}{n}\sin n\theta\right)\partial_\theta,$$
where $\theta\in[0,2\pi)$ parametrizes the circle and $n\in\N$. The completeness of the system $\RR=\{0,R\}$ means that for all the realizations $R_n$ there is a dense subset of the circle such that for all points of this subset the realization is locally equivalent to $R$ (or zero). Indeed, for all $\theta_0\neq \frac{k\pi}{n}$, $k\in\{0,1,\dots,n-1\}$, $R_n|_{\theta_0}$ is equivalent to $R$ through transformation
$$\theta\mapsto x=\frac{1}{n}\left(\cot\frac{n\theta}{2}-\cot\frac{n\theta_0}{2}\right).$$
\end{example}

Now, we are going to formulate a condition for a subsystem of a complete system of local realizations to stay complete.

\begin{definition}
Let $\bar\RR$ be a system of classes of local realizations of a Lie algebra $\g$ with respect to a group of automorphisms $A\subset\Aut\g$, $\bar\RR'\subset\bar\RR$. We will say that $\bar\RR'$ is a {\em sufficient subsystem} of $\bar\RR$ if for all classes $\bar R\in\bar\RR$ and all their representatives $R\in\bar R$ there exists $q\in\Dom R$ such that $\bar R|_q\in\bar\RR'$.
\end{definition}

\begin{lemma}
\label{L.equivchar}
Let $\bar\RR$ be a complete system of local realizations classes of a Lie algebra $\g$ with respect to $A\subset\Aut\g$. A subsystem $\bar\RR'\subset\bar\RR$ is complete if and only if $\bar\RR'$ is sufficient.
\end{lemma}
\begin{proof}
The left-right implication follows directly from the definition of completeness of the system~$\bar\RR'$.

Now let us take a realization $R$ on a manifold $M$. From completeness of $\bar\RR$ there is a point $p\in M$ such that $\bar R|_p\in\RR$. But from the definition of sufficient subsystem, taking $R$ as a representative of $R|_p\in\RR$, there exists $q\in M$ such that $R|_q\in\RR'$
\end{proof}

Again, we transfer this condition to the corresponding sets of representatives. For a set of local realizations $\RR$, its subset $\RR'\subset\RR$ is called {\em sufficient} with respect to a group $A\subset\Aut\g$ if $\bar\RR'$ is sufficient subsystem of $\bar\RR$.

\begin{proposition}
Let $\RR$ be a set of local realizations of a Lie algebra $\g$. A subset $\RR'\subset\RR$ is sufficient with respect to $A\subset\Aut\g$ if for all local realizations $R\in\RR$ at some $p\in\Dom R$ and for every neighborhood $U$ of $p$ there exists $q\in U$ and a realization $R'\in\RR'$ such that $R|_q$ is $A$-equivalent to~$R'$.
\end{proposition}
\begin{proof}
Take sets of local realizations $\RR'\subset\RR$ and denote $\bar\RR'$ and $\bar\RR$ the corresponding systems of $A$-classes of local realizations.

First we are going to prove the left-right implication, so assume that $\bar\RR'$ is sufficient. Take $R\in\RR$ a local realization at $p$, $U$ a neighborhood of $p$. Then $R|_U\in\bar R\in\bar\RR$. From completeness of $\bar\RR'$ there exists a $q\in\Dom R|_U=U$ such that $\overline{\left. (R|_U)\right|_q}=\bar R|_q\in\bar\RR'$, so there exists $R'\in\RR'$ such that $R|_q$ is $A$-equivalent to $R'$.

For the right-left implication, we have to prove that $\RR'$ is sufficient assuming the condition on the right hand side. So, take $\bar R\in\bar\RR$, $R\in\bar R$. Choose a representative $R_0\in\bar R$ that is contained in $\RR$, so $R_0$ is $A$-equivalent to $R$. This means that there exists $U_0\subset\Dom R_0$ and $U\subset\Dom R$, $\alpha\in A$ and $\Phi\colon U_0\to U$ a diffeomorphism such that $R|_U\circ\alpha=\Phi_*R_0|_{U_0}$. By assumption there exists $q_0\in U_0$ and $R'\in\RR'$ such that $R'$ is equivalent to $R_0|_{q_0}=(R_0|_{U_0})|_{q_0}$, which is also equivalent to $(R|_U)_q=R|_q$, where $q=\Phi(q_0)\in U$. Therefore, $R|_q$ is equivalent to $R'$ and hence $\bar R|_q\in\bar\RR'$.
\end{proof}



We claim that reasonable classification problem is to find a system of mutually inequivalent realizations that is complete. This is also more or less the classification problem that was solved in \cite{Popovych2003} for all Lie algebras of dimension less or equal to four, where the authors mentioned only Definition \ref{D.equiv} and remarked that they work only locally.

From Lemma \ref{L.regdense} it follows that system of all regular realizations is complete. Moreover, we can formulate the following lemma.

\begin{lemma}
\label{L.constrankonly}
Let $\RR$ be a complete system and $\RR'$ its subsystem containing all regular realizations of $\RR$. Then $\RR'$ is complete.
\end{lemma}
\begin{proof}
According to Lemma \ref{L.equivchar} we have to show that for every class of singular realizations of $\bar R\in\bar\RR$ and every representative $R\in\bar R$ there exists a $q\in\Dom R$ such that $\bar R|_q\in\bar\RR'$.

Choosing an open set $U\subset\Dom R$ of regular points of $R$ we can define restriction $R|_U$, which is a regular realization, and from completeness of $\RR$ there is point $q\in U$ such that $\bar R|_q\in\bar\RR$ and since it is regular, it belongs to $\bar\RR'$.
\end{proof}

Consequently, looking for a classification of realizations, one can deal with regular realizations only. In the following section, we present a simple algorithm for doing such classification.

Looking for complete system instead of classification of all regular realizations or even completely all realizations simplify the result as we are going to show on simple examples.

\begin{example}
Let us take the one dimensional Lie algebra $\g_1$ spanned by one element $e_1$. It is evident that any realization $R$ on any manifold $M$ is of the form $R_p(e_1)=X_p$, where $X$ is an arbitrary vector field on $M$. In particular, any local realization at $0\in\R^m$ is of the form
$$R_{x_1,\dots,x_m}(e_1)=\sum_{j=1}^m f_j(x_1,\dots,x_m)\partial_{x_j},$$
where $f_j$ are arbitrary smooth functions. If $f_j(0)\ne 0$ for some $j$, i.e. if the realization is regular, it can be transformed to $R'_{x_1,\dots,x_m}(e_1)=\partial_{x_1}$. However, the general classification problem does not have a reasonable solution because realizations of this form can be equivalent only if the sets $\{x\mid f(x)=0\}$ are diffeomorphic, so we would have to classify such functions, which would be very hard even in this very simple case.

Nevertheless, it actually is true that for every realization $R$ of $\g_1$ on some manifold $M$ there is a dense set of points $p\in M$ such that $R$ is locally equivalent to $\partial_{x_1}$ in $p$, so the realization $e_1\mapsto \partial_{x_1}$, together with the zero realization form a complete system.

And if we consider only analytic realizations, then the non-regular global realizations can be constructed as analytic continuations of the regular local ones.

For example, take $M$ to be a line $\R$. Then by applying $y=\Phi(x)=\exp x$, which maps $\R\to\R^+$, on the realization $R_x(e_1)=\partial_x$, $x\in\R$ we get realization
$$R_y'(e_1)=\partial_x y\partial_y=\e^x\partial_y=y\partial_y,\quad y\in\R^+,$$
whose analytic continuation is $R_x''(e_1)=x\partial_x$, $x\in\R$, which is a non-regular realization not equivalent to the former realization $R$.
\end{example}

\begin{example}
Now, consider two-dimensional Abelian Lie algebra $2\g_1$ and no automorphisms $A=\{\id\}=\Inn(2\g_1)$, which was already examined in Example \ref{E.2g1reg}, where we presented classification of all regular realizations. The result was that every regular realization with rank one in $m$ variables is equivalent to following
$$R^f(e_1)_{x_1,\dots,x_m}=\sin f(x_2,\dots,x_m)\,\partial_1,\quad R^f(e_1)_{x_1,\dots,x_m}=-\cos f(x_2,\dots,x_m)\,\partial_1,$$
where $f\colon\R^{m-1}\to\R$ are arbitrary functions defined locally in a neighborhood of zero. Here it can be shown, that the sufficient subsystem of those realizations are formed by following
$$R_1^a(e_1)_{x_1,\dots,x_m}=\partial_1,\quad R_1^a(e_1)_{x_1,\dots,x_m}=a\partial_1,$$
$$R_2^a(e_1)_{x_1,\dots,x_m}=\partial_1,\quad R_2^a(e_1)_{x_1,\dots,x_m}=(a+x_2)\partial_1,$$
where $a$ is a real parameter.
\end{example}

However, for Lie algebras of higher dimension we will already not be able to avoid families of realizations parametrized by functions.

\section{Sufficient subsystems of regular realizations}
\label{sec.complete}

In this section we present an algorithm for construction of complete system of local realizations.


First of all, we need to parametrize the subalgebra classes properly. That is, to find a decomposition of $\bar\SS_r$ to a finite disjoint union $\bigcup_i S_r^i(D_i)$, where $S_r^i\colon D_i\to\bar\SS_r$ is embedding of $D_i$ a domain in $\R^{s_i}$. Such a parametrization will be called \emph{proper} if for every smooth map $S\colon U\to\bar\SS_r$, where $U$ is a neighborhood of zero in $\R^s$, there exists $x\in U$ and its neighborhood $V$ such that $S(V)\subset S_r^i(D_i)$ for some $i$.

\begin{proposition}
\label{P.parametr}
Considering a proper parametrization, the regular local realizations corresponding to local $\bar\SS_r$-valued maps $S_r^i\circ F$, where $F\colon\R^{m-r}\to\R^{s_i}$, form a sufficient subsystem of all regular local realizations with rank $r$.
\end{proposition}
\begin{proof}
For a general map $S\colon\R^{m-r}\to\bar\SS_r$ we find the corresponding neighborhood $V$. Since $S(V)\subset S_r^i(D_i)$ for some $i$, there exists $F\colon\R^{m-r}\to\R^{s_i}$ such that $S=S_r^i\circ F$ on $V$.
\end{proof}

Finally, we can try to simplify the map $F$. We formulate the result in Lemma \ref{L.nontrans}. Since its formulation is rather complicated, we illustrate it on an example.

\begin{example}
Take a three-dimensional Abelian Lie algebra $3\g_1=\Span\{e_1,e_2,e_3\}$. There are no inner automorphisms. Every one-dimensional subspace is a subalgebra. Therefore, we can present the following proper parametrization of the set of one-dimensional subalgebras $\SS_2$ and compute the corresponding transitive realizations.
\begin{eqnarray*}
	\Span\{e_1\}&\qquad& e_1\mapsto0,\quad e_2\mapsto \partial_1,\quad e_3\mapsto\partial_2\\
	\Span\{e_2-ae_1\}&\qquad& e_1\mapsto\partial_1,\quad e_2\mapsto a\partial_1,\quad e_3\mapsto\partial_2\\
	\Span\{e_3-ae_2-be_1\}&\qquad& e_1\mapsto\partial_1,\quad e_2\mapsto \partial_2,\quad e_3\mapsto a\partial_1+b\partial_2
\end{eqnarray*}

According to Proposition \ref{P.parametr}, the following realizations form a sufficient subsystem of local realizations in $m$ variables with rank two
\begin{eqnarray*}
	&&e_1\mapsto0,\quad e_2\mapsto \partial_1,\quad e_3\mapsto\partial_2,\\
	&&e_1\mapsto\partial_1,\quad e_2\mapsto f(x_3,\dots,x_m)\partial_1,\quad e_3\mapsto\partial_2,\\
	&&e_1\mapsto\partial_1,\quad e_2\mapsto \partial_2,\quad e_3\mapsto f(x_3,\dots,x_m)\partial_1+g(x_3,\dots,x_m)\partial_2,
\end{eqnarray*}
where $f$ and $g$ are arbitrary local functions in $m-2$ variables. We will try to find a sufficient subsystem of the family of realizations in the last row.

Several cases can take place here. At first, if both $f$ and $g$ are constant in some neighborhood of zero, then we get a trivial extension of the original transitive realization only. Secondly, one of the functions may be locally constant at zero, while the other might not be. Then it means there is a point $(\epsilon_3,\dots,\epsilon_m)$ in a neighborhood of zero, where the first function, for example $f$, is locally constant equal to $a$, while the second function $g$ has a non-zero partial derivative with respect to some coordinate $x_i$, $i\ge 3$. Without lost of generality, we can assume $\left.\frac{\partial g}{\partial x_3}\right|_{(\epsilon_3,\dots,\epsilon_m)}\neq 0$ (otherwise we can change the order of coordinates in the first place). Now we present a change of coordinates
$$x_3\mapsto y_3:=g(x_3-\epsilon_3,\dots,x_m-\epsilon_m)-c,$$
$$x_i\mapsto y_i:=x_i-\epsilon_i,\quad i>3$$
where $c=g(\epsilon_3,\dots,\epsilon_m)$. In those coordinates, the realization taken in the neighborhood of $(\epsilon_3,\dots,\epsilon_m)$, which has new coordinates $y_i=0$, has the form
$$e_1\mapsto\partial_1,\quad e_2\mapsto\partial_2,\quad e_3\mapsto a\partial_1+(c+y_3)\partial_2.$$
Analogically, we obtain a realization
$$e_1\mapsto\partial_1,\quad e_2\mapsto\partial_2,\quad e_3\mapsto (c+y_3)\partial_1+b\partial_2.$$
Finally, both functions might not be constant. Then the function $f$ can again be transformed into $a+y_3$. The function $g$ might then depend only on $y_3$, so we get
$$e_1\mapsto\partial_1,\quad e_2\mapsto\partial_2,\quad e_3\mapsto (a+y_3)\partial_1+\tilde g(y_3)\partial_2$$
or it can depend on other variables as well (if $m\ge 4$), so it can be transformed into $y_4$, so we get
$$e_1\mapsto\partial_1,\quad e_2\mapsto\partial_2,\quad e_3\mapsto (a+y_3)\partial_1+(b+y_4)\partial_2.$$
\end{example}

\begin{lemma}
\label{L.nontrans}
Let $\bar\SS_r$ be the system of all Inn-classes of $\g$ subalgebras with codimension $r$. Let $D$ be an open domain in $\R^s$, $S\colon D\to\bar\SS_r$ an immersion. Then the system of realizations of $\g$ of the form
\begin{equation}
\label{eq.Treg2}
R(e_j)_{x_1,\dots,x_m}=R^{(S\circ F)(x_{r+1},\dots,x_m)}(e_j)_{x_1,\dots,x_r},
\end{equation}
where $F\colon\R^{m-r}\to D$ is a smooth function, and $R^{(S\circ F)(x_{r+1},\dots,x_m)}$ are transitive local realizations at $0\in\R^r$ corresponding to the subalgebra class $(S\circ F)(x_{r+1},\dots,x_m)$ chosen in a way that they smoothly depend on the coordinates, has a sufficient subsystem with respect to strong equivalence (so with respect to Inn-equivalence as well) consisting of the following local realizations at $0\in\R^m$
\begin{equation}
\label{eq.nontranfin}
R(e_j)_{x_1,\dots,x_m}=R^{(c_1+f_1(x_{r+1},\dots,x_m),\dots,c_s+f_s(x_{r+1},\dots,x_m))}(e_j)_{x_1,\dots,x_r},\
\end{equation}
where $(c_1,\dots,c_s)\in D$ are constant numbers and $f_j$ are local functions mapping $f_j(0,\dots,0)=0$ that are of the following form. Set $l_0=0$, then for all $j>0$ either $f_j(x_{r+1},\dots,x_m)=x_{l_j}$, where $l_j=l_{j-1}+1$, or $f_j(x_{r+1},\dots,x_m)$ depend only on first $l_j$ variables, where $l_j=l_{j-1}$. That is, either $f_j$ is equal to a ``new'' variable or it depends only on ``already used'' variables.
\end{lemma}
\begin{proof}
We have to find a suitable local diffeomorphism $\Psi\colon\R^{m-r}\to\R^{m-r}$ mapping a point in an arbitrarily small neighborhood of zero onto zero such that $F_j\circ\Psi=c_j+f_j$ on an even smaller neighborhood of this point.

Let us start with $F_1$. If it is locally constant at zero (i.e. there is a neighborhood of zero such that $F_1$ is constant on this neighborhood), then it does not need to be transformed. We just have to restrict ourselves on this neighborhood. Otherwise, there exists at every neighborhood of zero a point $x^{(1)}$ and an index $j\in\{1,\dots,m-r\}$ such that $\left.\frac{\partial F_1}{\partial x_{r+j}}\right|_{x^{(1)}}\neq 0$. Without loss of generality, assume $j=1$ (otherwise, apply diffeomorphism changing the order of variables at first). We can apply a diffeomorphism $\Psi_1$ mapping
$$x_{r+1}\mapsto y_{r+1}=F_1(x_{r+1}-x^{(1)}_{r+1},\dots,x_m-x^{(1)}_m)-c_1^{(1)},$$
$$x_{r+j}\mapsto y_{r+j}=x_{r+j}-x^{(1)}_{r+j},\quad j>1$$
where $c_1^{(1)}=F_1(x^{(1)}_{r+1},\dots,x^{(1)}_m)$, so $(F_1\circ\Psi_1)(y_{r+1},x_{r+2},\dots,x_m)=c_1^{(1)}+y_{r+1}$.

Then we proceed by induction. Assume, we have found a diffeomorphism $\Psi_k$, such that $(F_i\circ\Psi_k)(x_{r+1},\dots,x_m)=c^{(k)}_i+f_i^{(k)}(x_{r+1},\dots,x_m)$ for all $i\le k$ on some neighborhood of zero. Then we study $F_{k+1}\circ\Psi_k$. If it depends only on $x_{r+1},\dots,x_{l_k-1}$, then nothing has to be done, so $\Psi_{k+1}=\Psi_k$. If it has non-zero partial derivative with respect to $x_j$, $j\ge l_k$ in a point $x^{(k+1)}$ arbitrarily close to $x^{(k)}$, without loss of generality let $j=l_k$, then we can introduce a diffeomorphism $\tilde\Psi_{k+1}$ sending
$$x_{r+k+1}\mapsto y_{r+k+1}=F_{k+1}(x_{r+1}-x^{(k+1)}_{r+1},\dots,x_m-x^{(k+1)}_m)-c^{(k+1)}_{k+1},$$
$$x_{r+j}\mapsto y_{r+j}=x_{r+j}-x^{(k+1)}_{r+j},\quad j\neq k,$$
where $c^{(k+1)}_{k+1}=F_{k+1}(x^{(k+1)}_{r+1},\dots,x^{(k+1)}_m)$. Then if we define $\Psi_{k+1}:=\tilde\Psi_{k+1}\circ\Psi_k$, we have $(F_{k+1}\circ\Psi_{k+1})(y_{r+1},\dots,y_m)=c^{(k+1)}_{k+1}+y_{k+1}$. The form of $F_i\circ\Psi_{k+1}$ for $i<k+1$ remains the same, only the constants $c^{(k)}_i$ change to new $c^{(k+1)}_i$ because of the translation.
\end{proof}

The functions $\vec c+\vec f$ in Lemma \ref{L.nontrans} were constructed in such a way that they are mutually inequivalent with respect to change of coordinates, so the set of realizations constructed by this lemma contains mutually Inn-inequivalent realizations.
%

Finally, we can summarize the algorithm for construction of complete system of realizations into the following theorem.

\begin{theorem}
\label{T.main}
Let
$$\bar\SS_m=\bigcup_i\bar S_m^i(D_i),\quad D_i\subset\R^{s_i}$$
be a proper parametrization of classes of $\g$ subalgebras with codimension $m$. Let $\SS_m=\bigcup_i S_m^i(D_i)$ be such set of their representatives that $S_m^i$ are smooth. Let $\TT_m=\bigcup_i R_m^{i\,(D_i)}$ be the corresponding local transitive realizations. Let $\RR_m^i$ be the sets of regular realizations constructed as in Lemma \ref{L.nontrans} from $s_i$-parameter sets $R_m^{i\,(D_i)}$ of transitive realizations. Then the system of local realizations $\RR=\bigcup_m\left(\TT_m\cup\bigcup_i\RR_m^i\right)$ is complete with respect to inner automorphisms and contains mutually Inn-inequivalent local realizations.
\end{theorem}

\begin{remark}
\label{R.completestrong}
If we are doing classification with respect to strong equivalence, all we have to do in the end is to apply the inner automorphisms to the complete system we have found. This essentially mean putting back the parameters we have eliminated when doing the classification with respect to inner automorphisms.
\end{remark}

\begin{remark}
\label{R.completeAut}
If we are doing classification with respect to all automorphisms, we make use of Remark \ref{R.outer}. It follows that if we have a parametrization of Inn-classes of subalgebras $S(a_1,\dots,a_s)$ in such a way that for first $t$ parameters parametrize Aut-equivalent subalgebras, that is, subalgebras $S(a_1,\dots,a_t,a_{t+1},\dots,a_s)$ and $S(\tilde a_1,\dots,\tilde a_t,a^{t+1},\dots,a_s)$ are always Aut-equivalent, then realizations constructed in Lemma \ref{L.nontrans} that differ only in constants $c_1,\dots,c_t$ are also $A$-equivalent.
\end{remark}

Finally, we are going to illustrate the presented algorithm for construction of complete system of realizations, on a more complicated example. We illustrate classification with respect to strong, Inn-, and Aut-equivalnce.

\begin{example}
Let us consider a Lie algebra $\g=\g_{2.1}\oplus 2\g_1=\Span\{e_1,e_2,e_3,e_4\}$, $[e_1,e_2]=e_1$. We are going to find a complete system of realizations with rank three. The groups of automorphisms expressed in the basis $(e_1,e_2,e_3,e_4)$ are following
\begin{eqnarray}
\Inn\g&=&\left\{\begin{pmatrix}\tilde t_1&t_2&0&0\\0&1&0&0\\0&0&1&0\\0&0&0&1\end{pmatrix}\mathrel{\Bigg|}\tilde t_1\in\R^+,t_2\in\R\right\}
\label{eq.exin}\\
\Aut\g&=&\left\{\begin{pmatrix}t_1&t_2&0&0\\0&1&0&0\\0&t_3&t_5&t_6\\0&t_4&t_7&t_8\end{pmatrix}\mathrel{\Bigg|}t_1,t_2,t_3,t_4,t_5,t_6,t_7,t_8\in\R,\;t_1\neq 0,\;t_5t_8-t_6t_7\neq 0\right\}
\label{eq.exaut}
\end{eqnarray}

\begin{table}[t]
\begin{center}
\begin{tabular}{lllll}
\hline
                &  Subalgebras             &  Inn-classes             &  Aut-classes          &  Is ideal?\\\hline
$\h_1$          &  $e_1$                   &  $e_1$                   &  $e_1$                &  Yes.\\
$\h_2^c$        &  $e_3+ce_1$              &  $e_3+\sgn(c)e_1$        &  $e_4+\eta_ce_1$      &  If $c=0$.\\
$\h_3^{a,c}$    &  $e_4+ae_3+ce_1$         &  $e_4+ae_3+\sgn(c)e_1$   &  $e_4+\eta_ce_1$      &  If $c=0$.\\
$\h_4^{a,b,c}$  &  $e_2+ae_3+be_4+ce_1$    &  $e_2+ae_3+be_4$         &  $e_2$  &  No.\\\hline
\end{tabular}
\end{center}
\caption{One-dimensional subalgebras of $\g$}
\label{t.subalgebras}
\end{table}

\begin{table}[t]
\begin{center}
\begin{tabular}{llll}
\hline
                &  Subalgebras             &  Realizations                                          &  Is faithful?\\\hline
$\h_1$          &  $e_1$                   &  $0$, $\D_1$, $\D_2$, $\D_3$                           &  No.\\
$\h_2^c$        &  $e_3+ce_1$              &  $\D_1$, $x_1\D_1+\D_2$, $-c\e^{x_2}\D_1$, $\D_3$       &  If $c\neq 0$.\\
$\h_3^{a,c}$    &  $e_4+ae_3+ce_1$         &  $\D_1$, $x_1\D_1+\D_2$, $\D_3$, $-c\e^{x_2}\D_1-a\D_3$  &  If $c\neq 0$.\\
$\h_4^{a,b,c}$  &  $e_2+ae_3+be_4+ce_1$    &  $\D_1$, $(x_1-c)\D_1-a\D_2-b\D_3$, $\D_2$, $\D_3$     &  Yes.\\\hline
\end{tabular}
\end{center}
\caption{Classification of local transitive realizations of $\g$ in three variables}
\label{t.trans}
\end{table}

Of course, all one-dimensional subspaces are also one-dimensional subalgebras. In Table \ref{t.subalgebras}, we classify the one-dimensional subalgebras with respect to automorphisms. In the first column, parametrization of all subalgebras is presented. In the second column we choose an Inn-representative for each subalgebra. In the third column we choose the Aut-representatives. We denote $\eta_a=\sgn|a|$. In the last column we express a condition for the subalgebra to be an ideal.

The Shirokov's computation leads to classification of transitive realizations with respect to strong equivalence listed in Table \ref{t.trans}. In the second column, a realization corresponding to subalgebra in the first column is presented. The entry consists of images of the basis elements $e_1$, $e_2$, $e_3$, and $e_4$. In the last column, the negation of last column of Table \ref{t.subalgebras} expresses a condition for the realization to be faithful.

The classification with respect to inner or all automorphisms is obtained easily by substituting the parameters of subalgebras by parameters of the chosen representatives. By doing so in case of inner automorphisms and applying Lemma \ref{L.nontrans}, we get a complete system of realizations of $\g$ with respect to inner automorphisms and list it in Table \ref{t.inn}. The function $f\colon\R\to\R$ is an arbitrary local function satisfying $f(0)=0$. All realizations are Inn-inequivalent.

\begin{table}[ht]
\begin{center}
\begin{tabular}{ll}
\hline
Subalgebra                 &  Realization                                        \\\hline
$\h_1$                     &  $0$, $\D_1$, $\D_2$, $\D_3$                        \\
$\h_2^0$                   &  $\D_1$, $x_1\D_1+\D_2$, $0$, $\D_3$                 \\
$\h_2^1$                   &  $\D_1$, $x_1\D_1+\D_2$, $-\e^{x_2}\D_1$, $\D_3$     \\
$\h_2^{-1}$                &  $\D_1$, $x_1\D_1+\D_2$, $\e^{x_2}\D_1$, $\D_3$      \\
$\h_3^{a,0}$               &  $\D_1$, $x_1\D_1+\D_2$, $\D_3$, $-a\D_3$             \\
$\h_3^{a+x_4,0}$           &  $\D_1$, $x_1\D_1+\D_2$, $\D_3$, $-(a+x_4)\D_3$             \\
$\h_3^{a,1}$               &  $\D_1$, $x_1\D_1+\D_2$, $\D_3$, $-\e^{x_2}\D_1-a\D_3$\\
$\h_3^{a+x_4,1}$           &  $\D_1$, $x_1\D_1+\D_2$, $\D_3$, $-\e^{x_2}\D_1-(a+x_4)\D_3$\\
$\h_3^{a,-1}$              &  $\D_1$, $x_1\D_1+\D_2$, $\D_3$, $\e^{x_2}\D_1-a\D_3$ \\
$\h_3^{a+x_4,-1}$          &  $\D_1$, $x_1\D_1+\D_2$, $\D_3$, $\e^{x_2}\D_1-(a+x_4)\D_3$ \\
$\h_4^{a,b,0}$             &  $\D_1$, $x_1\D_1-a\D_2-b\D_3$, $\D_2$, $\D_3$      \\
$\h_4^{a,b+x_4,0}$         &  $\D_1$, $x_1\D_1-a\D_2-(b+x_4)\D_3$, $\D_2$, $\D_3$      \\
$\h_4^{a+x_4,b+f(x_4),0}$  &  $\D_1$, $x_1\D_1-(a+x_4)\D_2-(b+f(x_4))\D_3$, $\D_2$, $\D_3$      \\
$\h_4^{a+x_4,b+x_5,0}$     &  $\D_1$, $x_1\D_1-(a+x_4)\D_2-(b+x_5)\D_3$, $\D_2$, $\D_3$      \\\hline
\end{tabular}
\end{center}
\caption{Complete system of realizations of $\g$ with rank three with respect to inner automorphisms}
\label{t.inn}
\end{table}

Now, let us make the classification with respect to strong equivalence. As we mentioned in Remark \ref{R.completestrong}, we just have to put back the eliminated parameters. That is, instead of listing realizations corresponding, for example, to $\h_4^{a,b+x_4,0}$, we list realizations corresponding to all $\h_4^{a,b+x_4,c}$. The result is in Table \ref{t.strong}. All realizations are strongly inequivalent.

\begin{table}[ht]
\begin{center}
\begin{tabular}{ll}
\hline
Subalgebra                 &  Realization                                        \\\hline
$\h_1$                     &  $0$, $\D_1$, $\D_2$, $\D_3$                        \\
$\h_2^c$                   &  $\D_1$, $x_1\D_1+\D_2$, $-c\e^{x_2}\D_1$, $\D_3$     \\
$\h_3^{a,c}$               &  $\D_1$, $x_1\D_1+\D_2$, $\D_3$, $-c\e^{x_2}\D_1-a\D_3$\\
$\h_3^{a+x_4,c}$           &  $\D_1$, $x_1\D_1+\D_2$, $\D_3$, $-c\e^{x_2}\D_1-(a+x_4)\D_3$\\
$\h_4^{a,b,c}$             &  $\D_1$, $(x_1-c)\D_1-a\D_2-b\D_3$, $\D_2$, $\D_3$      \\
$\h_4^{a,b+x_4,c}$         &  $\D_1$, $(x_1-c)\D_1-a\D_2-(b+x_4)\D_3$, $\D_2$, $\D_3$      \\
$\h_4^{a+x_4,b+f(x_4),c}$  &  $\D_1$, $(x_1-c)\D_1-(a+x_4)\D_2-(b+f(x_4))\D_3$, $\D_2$, $\D_3$      \\
$\h_4^{a+x_4,b+x_5,c}$     &  $\D_1$, $(x_1-c)\D_1-(a+x_4)\D_2-(b+x_5)\D_3$, $\D_2$, $\D_3$      \\\hline
\end{tabular}
\end{center}
\caption{Complete system of realizations of $\g$ with rank three with respect to strong equivalence
}
\label{t.strong}
\end{table}

Finally, let us do the classification with respect to all automorphisms. As we mentioned in Remark \ref{R.completeAut}, the first step is to remove the unneccessary parameters. For example, all subalgebras $\h_3^{a,c}$ for $c\neq 0$ are Aut-equivalent to $\h_3^{0,1}$. Therefore, realizations corresponding to $\h_3^{a,c}$ and $\h_3^{a+x_4,c}$ are equivalent to the realizations corresponding to $\h_3^{0,1}$ and $\h_3^{x_4,1}$. The result of such a process is listed in Table \ref{t.aut}.  In the last column, we list the number of the realization in the classification \cite{Popovych2003}, p.~7345 (only in case it is faithful).

\begin{table}[ht]
\begin{center}
\begin{tabular}{lll}
\hline
Subalgebra                 &  Realization                                              &  No. in \cite{Popovych2003}\\\hline
$\h_1$                     &  $0$, $\D_1$, $\D_2$, $\D_3$                              &    \\
$\h_3^{0,0}$               &  $\D_1$, $x_1\D_1+\D_2$, $\D_3$, $0$                      &    \\
$\h_3^{x_4,0}$             &  $\D_1$, $x_1\D_1+\D_2$, $\D_3$, $-x_4\D_3$               &  7  \\
$\h_3^{0,1}$               &  $\D_1$, $x_1\D_1+\D_2$, $\D_3$, $-\e^{x_2}\D_1$          &  6  \\
$\h_3^{x_4,1}$             &  $\D_1$, $x_1\D_1+\D_2$, $\D_3$, $-\e^{x_2}\D_1-x_4\D_3$  &  5  \\
$\h_4^{0,0,0}$             &  $\D_1$, $x_1\D_1$, $\D_2$, $\D_3$                        &  4  \\
$\h_4^{0,x_4,0}$           &  $\D_1$, $x_1\D_1-x_4\D_3$, $\D_2$, $\D_3$                &  3  \\
$\h_4^{x_4,f(x_4),0}$      &  $\D_1$, $x_1\D_1-x_4\D_2-f(x_4)\D_3$, $\D_2$, $\D_3$     &  3  \\
$\h_4^{x_4,x_5,0}$         &  $\D_1$, $x_1\D_1-x_4\D_2-x_5\D_3$, $\D_2$, $\D_3$        &  2  \\\hline
\end{tabular}
\end{center}
\caption{Complete system of realizations of $\g$ with rank three with respect to all automorphisms}
\label{t.aut}
\end{table}

The realization corresponding to $\h_4^{0,x_4,0}$ is Aut-equivalent to realization $\h_4^{x_4,0,0}$. Furthermore, realizations corresponding to $\h_4^{x_4,f(x_4),0}$ may be, for different $f$, also equivalent. All other realizations listed in \ref{t.aut} are Aut-inequivalent.

To find the condition for the realizations corresponding to $\h_4^{x_4,f(x_4),0}$ to be equivalent, we have to examine, how automorphisms act on this subalgebra-valued function. The equation
$$\alpha_{t_1,t_2,t_3,t_4,t_5,t_6,t_7,t_8}(\h_4^{x_4,f(x_4),0})=\h_4^{\tilde x_4,\tilde f(\tilde x_4),0},$$
where $\alpha$ is an automorphism of $\g$ parametrized as in \eqref{eq.exaut}, together with conditions $\tilde x_4(x_4=0)=0$ and $f(0)=0$, leads to constraint $t_2,t_3,t_4=0$ and finally is equivalent to
$$\tilde x_4=t_5x_4+t_6f(x_4),\quad\tilde f(\tilde x_4)=t_7x_4+t_8f(x_4),$$
Therefore, realizations corresponding to $\h_4^{x_4,f(x_4),0}$ and $\h_4^{x_4,\tilde f(x_4),0}$, where $f(0)=\tilde f(0)=0$, are equivalent if and only if
$$\tilde f(t_5x+t_6f(x))=t_7x+t_8f(x),$$
where $t_5t_8-t_6t_7\neq 0$ (cf. \cite{Popovych2003} p.~7351).
\end{example}

\subsection*{Acknowledgments}

The authors acknowledge support of the Czech Technical University in Prague, namely the projects SGS15/215/OHK4/3T/14 and SGS16/239/OHK4/3T/14. We thank M. Nesterenko for valuable discussions and comments.


\begin{thebibliography}{10}

\bibitem{Blattner1969}
R. J. Blattner.
\newblock Induced and Produced Representations of Lie Algebras.
\newblock \textit{Trans. Amer. Math. Soc.}, 144, 457--474, 1969.





\bibitem{draisma2012}
J. Draisma.
\newblock Transitive Lie Algebras of Vector Fields: An Overview.
\newblock {\em Qual. Theory Dyn. Syst.}, 11: 39--60, 2012.

\bibitem{draisma2002}
Draisma, J.
On a conjecture of Sophus Lie.
In: \textit{Differential Equations and the Stokes Phenomenon.}
World Scientific, Singapore, 65--87, 2002

\bibitem{Guillemin1964}
V. W. Guillemin, and S. Sternberg.
\newblock An algebraic model of transitive differential geometry.
\newblock \textit{Bull. Amer. Math. Soc.}, 70, 16–47, 1964.

\bibitem{Jafarpour}
S. Jafarpour, and A. Lewis.
\textit{Time-Varying Vector Fields and Their Flows}.
Springer International Publishing, 2014

\bibitem{lie1893}
S. Lie, and F. Engel
\textit{Theorie der Transformationsgruppen}, Vol.1--3.
Leipzig, 1888, 1890, 1893.

\bibitem{shirokov2013}
A.~A. Magazev, V.~V. Mikheyev, and I.~V. Shirokov.
\newblock Computation of composition functions and invariant vector fields in
  terms of structure constants of associated {L}ie algebras.
\newblock {\em SIGMA}, 11:066, 17 pages, 2015.

\bibitem{Munkers}
James R. Munkers.
\textit{Topology.}
Upper Saddle River: Prentice-Hall, 2000.

\bibitem{Nesterenko2015}
M.~Nesterenko.
\newblock The {P}oincar\'e algebras $p(1,1)$ and $p(1,2)$: realizations and
  deformations.
\newblock {\em J.~Phys. Conf. Ser.}, 621:012009, 2015.

\bibitem{Nesterenko2016}
M. Nesterenko, S. Po\v sta, and O. Vaneeva
\newblock Realizations of Galilei algebras
\newblock {\em J.~Phys.~A: Math. Theor.},  49(11), 115203, 26 pages, 2016.



\bibitem{Popovych2003}
R.~O. Popovych, V.~M. Boyko, M.~O. Nesterenko, and M.~W. Lutfullin.
\newblock Realizations of real low-dimensional {L}ie algebras.
\newblock {\em J.~Phys.~A}, 36(26):7337--7360, 2003.
\newblock See math-ph/0301029 for the revised and extended version.

\bibitem{Spichak}
S.~V. Spichak. Preliminary Classification of Realizations of Two-Dimensional Lie Algebras of Vector Fields on a Circle.
In: \textit{Sixth Workshop Group Analysis of Differential Equations and Integrable Systems},
212--218, 2013

\bibitem{shirokov1997}
I.~V. Shirokov.
\newblock Construction of {L}ie algebras of first-order differential operators.
\newblock {\em Izv. Vyssh. Uchebn. Zaved. Fiz.}, 40(6):25--32, 1997.



\bibitem{basarab-horwath&lahno&zhdanov2001}
Basarab-Horwath P., Lahno V. and Zhdanov R.,
The structure of Lie algebras and the classification problem for partial differential equations,
{\it Acta Appl. Math.}, 2001, V.69, N~1, 43--94,  math-ph/0005013.

\bibitem{BourliouxCyr-GagnonWinternitz}
Bourlioux A., Cyr-Gagnon C. and Winternitz P.,
Difference schemes with point symmetries and their numerical tests,
\textit{J. Phys. A: Math. Gen.}, 2006, V.39, 6877--6896, math-ph/0602057.

\bibitem{Carinena}
Carinena J.F., Grabowski J. and Marmo G.,
Some physical applications of systems of differential equations admitting a superposition rule,
{\it Rep. Math. Phys.}, 2001, V.48, 1–2, 47--58.

\bibitem{Levine}
Levine R.D.,
Lie algebraic approach to molecular structure and
dynamics, in Mathematical Frontiers in Computational Chemical Physics,
ed. D.G. Truhlar, IMA Volumes in Mathematics and its Applications,
Vol.15, New York, Springer-Verlag, 1988, 245--261.

\bibitem{Lie1}
Lie~S.,
\"Uber Differentiation, {\it Math.~Ann.}, 1884, V.24, 537--578;  Gesammetle Abhandlungen, Vol.\ 6,
Leipzig, B.G. Teubner, 1927, 95--138.

\bibitem {lie1883} Lie S., Klassification und Integration von gew\"ohnlichen
Differentialgleichungen zwischen $x$, $y$,
die eine Gruppe von Transformationen gestatten, {\it Arch. Math. Naturv.},
1883, V.9, 371--393.

\bibitem{lie1888_1890_1893} Lie S., Theorie der Transformationsgruppen, Vol.1--3,
Leipzig, B.G. Teubner, 1888, 1890, 1893 (see pages 57--73, 713--732 of the Vol.3).


\bibitem{petrov1966}  Petrov A.Z.,
New methods in general relativity,
Moscow, Nauka, 1966 (in Russian).

\bibitem{Popovych&Boyko}
Popovych R. and Boyko V.,
Differential invariants and application to Riccati-type systems,
{\it Proceedings of Institute of Mathematics},
Kyiv, 2002, V.43, Part 1, 184--193, math-ph/0112057.


\bibitem{Olver0}
Olver P.J.,
Applications of Lie groups to differential equations, New York, Springer-Verlag, 1993.

\bibitem{Olver1}
Olver P.J.,
Differential invariants and invariant differential equations,
{\it Lie Groups Appl.}, 1994, V.1, 177--192.

\bibitem{Olver2}
Olver P.J.,
Equivalence, invariants, and symmetry, Cambridge University Press, 1995.

\bibitem{Ovsyannikov}
Ovsiannikov~L.V.,
Group analysis of differential equations,
New York, Academic Press, 1982.

\bibitem{PateraSharpWinternitz1976}
Patera J., Sharp R.T., Winternitz P. and Zassenhaus H.,
Invariants of real low dimension Lie algebras,
\textit{J. Math. Phys.}, 1976, V.17, 986--994,

\bibitem{Shnider_Winternitz1984}
Shnider S. and Winternitz P.,
Nonlinear equations with superposition principles and the theory of transitive primitive Lie algebras,
{\it Lett. Math. Phys.}, 1984, V.8, 69--78.

\end{thebibliography}
\end{document}